\documentclass[AMA,STIX2COL]{WileyNJD-v2}
\usepackage{algorithm}
\usepackage{algorithmicx}
\usepackage{xcolor}

\hypersetup{
	colorlinks=false,
	pdfborder={0 0 0},
}

\articletype{Original Research}

\received{}
\revised{}
\accepted{}

\begin{document}

\title{Data-driven dual-loop control for platooning mixed human-driven and automated vehicles}

\author{Jianglin Lan} 

\authormark{Jianglin Lan}

\address{\orgdiv{James Watt School of Engineering}, \orgname{University of Glasgow}, \orgaddress{\state{Glasgow G12 8QQ}, \country{United Kingdom}}}

\corres{Jianglin Lan, James Watt School of Engineering, University of Glasgow, Glasgow G12 8QQ, United Kingdom. \email{Jianglin.Lan@glasgow.ac.uk}}	

\fundingInfo{Leverhulme Trust, Grant/Award Number: ECF-2021-517.}

\abstract[Summary]{
This paper considers controlling automated vehicles (AVs) to form a platoon with human-driven vehicles (HVs) under consideration of unknown HV model parameters and propulsion time constants.
The proposed design is a data-driven dual-loop control strategy for the ego AVs, where the inner loop controller ensures platoon stability and the outer loop controller keeps a safe inter-vehicular spacing under control input limits.
The inner loop controller is a constant-gain state feedback controller solved from a semidefinite program (SDP) using the online collected data of platooning errors. 
The outer loop is a model predictive control (MPC) that embeds a data-driven internal model to predict the future platooning error evolution. 
The proposed design is evaluated on a mixed platoon with a representative aggressive reference velocity profile, the SFTP-US06 Drive Cycle. The results confirm efficacy of the design and its advantages over the existing single loop data-driven MPC in terms of platoon stability and computational cost. 
}



\maketitle

\section{Introduction}
Cooperative adaptive cruise control (CACC) powered by vehicle-to-vehicle (V2V) communications enables the vehicles to form a platoon, i.e., travelling in line at the same speed with safe inter-vehicular spacing. The deployment of vehicle platoons can potentially improve traffic capacity, safety and fuel consumption \cite{van2006impact,jia2015survey}. 
Currently there are a large volume of effective CACC strategies for pure automated vehicles (AVs) platoons \cite{li2017dynamical,guanetti2018control,rahman2018longitudinal}. 
However, in the near future it is most likely to see the co-existence of AVs and human-driven vehicles (HVs) driving on the same roads \cite{di2020survey}. 
Compared to the control of pure AVs platoons, CACC for mixed AVs and HVs platoons will experience the additional challenges:
(i) The mixed vehicle platoon is not fully controllable as the HVs are controlled by human drivers instead of automatic algorithms like the AVs, and  
(ii) human driving behaviours have intrinsic uncertainty and unpredictability, easily causing traffic congestion \cite{sugiyama2008traffic} and oscillation \cite{qin2019experimental}.  
However, the existing platooning strategies for pure AVs are not for this purpose, raising the need for new methods tailored for mixed vehicle platoons.

A potential way to address the above mentioned challenges is 
taking into account of the HVs behaviours when controlling the AVs so as to ensure safety and robustness of the mixed vehicle platoon. 
The optimal velocity (OV) model \cite{orosz2010traffic} has been shown to be capable of characterising how a HV follows its preceding vehicle in a mixed platoon  \cite{zheng2020smoothing,giammarino2020traffic,hajdu2019robust,chen2020mixed,wang2020leading,feng2019robust}. The OV model has been used to design CACC strategies for mixed platoons in several works.  
The works \cite{zheng2020smoothing,giammarino2020traffic} propose state feedback control for an AV to smooth the mixed traffic flow in ring roads;
Robust control is adopted for a more generalised mixed platoon \cite{hajdu2019robust}; 
Optimal control is used for mixed platoons on straight roads \cite{wang2020leading} or at signalised intersections \cite{chen2020mixed};
Tube-based model predictive control (MPC) has also been developed for mixed platoons under vehicle physical constraints \cite{feng2019robust}.
The above works rely on a common assumption of known OV models, which is too restrictive as accurately modelling HVs behaviours is difficult \cite{di2020survey}. A CACC design not relying on known HV parameters is thus more appealing.

By leveraging the adaptive dynamic programming theory \cite{lewis2012reinforcement}, 
data-driven CACC have been designed for mixed vehicle platoons under consideration of control constraints \cite{gao2016data} and reaction delays in HVs \cite{huang2020learning}.
However, their designs cannot ensure platoon safety and are not robust to changes in the leader velocity and uncertainties in the HVs behaviours. 
Data-driven MPC methods have been proposed in the literature \cite{Mahbub+23,Zhan+22,Wang+23i,Wang+22,Li+23p,Lan+21d} to overcome these drawbacks.
A data-driven receding horizon control is proposed in the work \cite{Mahbub+23} by using recursive least squares to estimate the HVs models. However, their design aims for a special type of mixed vehicle platoons where an AV is controlled to lead a set of HVs. 	
The work \cite{Zhan+22} uses the Koopman operator theory with neural networks to build a linear vehicle platoon model for centralised and distributed MPC designs. However, their method needs sufficient offline platoon data for neural network training and online model update when the modelling accuracy does not meet the prescribed threshold.
The work \cite{Wang+23i} proposes a MPC design based on a HV model that combines a nominal model (a second-order transfer function) and a Gaussian process model. Their work requires offline dataset to train the Gaussian process model and focuses on the special mixed platoon formation with only one HV driving behind a set of AVs. 
The work \cite{Wang+22} proposes the DeeP-LCC method, where the MPC design relies on the platoon behaviour predicted by a Hankel matrix built from collected platoon data. However, their work assumes that the future prediction of the disturbance (i.e., the lead vehicle velocity changes) is zero, which is restrictive. This Hankel matrix-based MPC design is further developed in the work \cite{Li+23p} to optimise the holistic energy consumption of the mixed platoon.
The work \cite{Lan+21d} uses zonotope reachability analysis to generate an over-approximated prediction of the future mixed platoon behaviours. However, the use of reachability set as prediction is conservative and the involvement of matrix zonotope operations makes the MPC optimisation problems to be solved online computationally expensive.

This paper aims for a new data-driven CACC strategy for mixed vehicle platoons with unknown parameters for HV models and unknown propulsion time constants for all vehicles. This paper mainly contributes in three aspects:
\begin{enumerate}
	\item[1)] The proposed CACC strategy has a dual-loop structure, where the inner loop is a state feedback controller ensuring asymptotic stability of the platoon and the outer loop is a MPC controller ensuring satisfaction of the safety/input constraints and robustness against leader velocity changes. The outer loop incorporates an internal model that is updated efficiently in real-time to capture the time-varying HVs' behaviours and the model disturbances, thus addressing the limitations in the existing designs \cite{Zhan+22,Lan+21d,Wang+22}.
	\item[2)] Data-driven methods are proposed to design both the inner and outer loop controllers. The inner loop control design integrates the recent advancement of data-driven control theory \cite{DeTesi19,Van+20a,Van+20b,Persis+22} into a dual-loop system structure, where the gain is solved from a data-driven semidefinite program (SDP) and to be refined by the MPC policy for constraint satisfaction. The MPC design takes inspiration from the existing concept of offset free MPC for known system models \cite{Maeder+09,PannocchiaBemporad07}, but with a newly data-driven design strategy. 
	\item[3)] Simulation results of a mixed platoon under a representative realistic driving scenario, the SFTP-US06 Drive Cycle, show that the proposed design is computationally cheaper and can maintain a more stable platoon than the existing data-driven MPC \cite{Lan+21d}. 
\end{enumerate}

The rest of this paper is structured as follows:
Section \ref{sec:problem statement} describes the design problem. 
Section \ref{sec: inner loop} presents the data-driven inner loop control design, followed by the outer loop control design in Section~\ref{sec: outer loop}. Section \ref{sec:simulation} reports the simulation results and Section \ref{sec:conclusion} draws the conclusions.

\textbf{Notations}: 
The symbol $[a,b]$ denotes the sequence of integers from $a$ to 
$b$. $\mathrm{diag}(V_1, \cdots, V_n)$ denotes a block diagonal matrix with $V_1, \cdots, V_n$ being the main diagonals. $\mathbf{1}_n$ is a $n$-dimensional column of ones. 
$I$ and $\mathbf{0}$ are identity and zero matrices, respectively, whose dimensions are known from the context and omitted unless they are necessary to be given. $\star$ indicates symmetry in a block matrix. s.t. is the abbreviation for subject to.

\section{Platoon modelling and control problem} \label{sec:problem statement}
This paper studies the generic mixed vehicle platoon with the formation illustrated in Figure \ref{fig1}. The platoon consists of $n+1$ vehicles, two AVs at the very front (AV 0) and rear (AV $n$) and $n-1$ HVs elsewhere. All the vehicles can communicate with each other through the V2V wireless network. The assumption of V2V communications between HVs and AVs has been commonly made in the existing CACC for mixed platoons, including the model-based control methods \cite{zheng2020smoothing,hajdu2019robust,chen2020mixed,wang2020leading,feng2019robust} and the data-driven control methods \cite{gao2016data,huang2020learning,Mahbub+23,Zhan+22,Wang+22,Lan+21d}. 
Setting an AV as the leader is necessary to ensure formation of the mixed platoon. The leader will also be used to generate platooning data for the proposed data-driven CACC for AV $n$. As shown in the work \cite{Lan+21d}, the CACC to be designed for the mixed platoon in Figure \ref{fig1} can be directly used for general mixed vehicle platoons by splitting them into several sub-platoons with the formation in Figure \ref{fig1}.

\begin{figure}[t]
	\centering
	\includegraphics[width=\columnwidth]{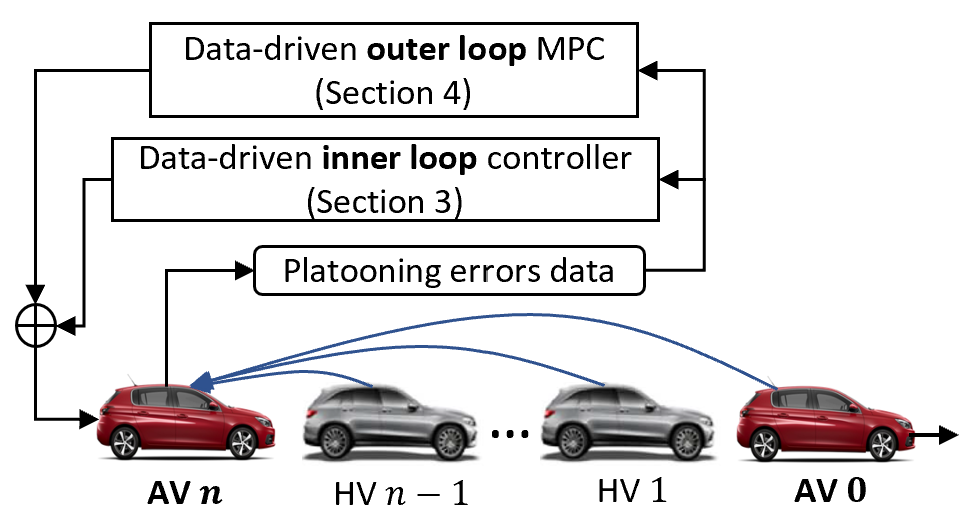} 
	\vspace{-3mm}
	\caption{The mixed platoon formation and proposed dual-loop CACC framework.}
	\label{fig1}
\end{figure}

Similar to the exiting vehicle platooning literature, this paper focuses on designing the longitudinal control for AV $n$. 
The $i$-th AV, $i = 0, n$, has the longitudinal dynamic model
\begin{subequations}\label{virtual leader dynamics}
	\begin{align}
		\dot{p}_i &= v_i, \\
		\dot{v}_i &= a_i, \\
		\dot{a}_i &= \frac{1}{\tau_i} (u_i - a_i),
	\end{align}
\end{subequations}
where $p_i$, $v_i$, $a_i$ and $u_i$ represent the position, velocity, acceleration and control command of the vehicle, respectively. $\tau_i$ denote the propulsion time constants, where $\tau_0$ of the leader is unknown but $\tau_n$ of the ego AV is known.

The leader is controlled by an existing controller $u_0$ such as a MPC controller \cite{lan2020min} and can follow the given longitudinal velocity reference $v^*$ well. 
The AV $n$ is controlled by $u_n$ to follow the leader and maintain a safe gap $d_\mathrm{safe}$ from HV $n-1$. For simplification, a constant $d_\mathrm{safe}$ is used in this paper but the proposed design is extendable to the case of time-varying $d_\mathrm{safe}$.
The platooning errors of AV $n$ are defined as $\tilde{h}_n = h_n - d_\mathrm{safe}$ and $\tilde{v}_n = v_n - v^*$, where $h_n = p_{n-1} - p_n$ is the gap between vehicles $n-1$ and $n$. 
The platooning error dynamics of AV $n$ are derived from \eqref{virtual leader dynamics} and given as
\begin{equation} \label{AV error sys}
\dot{x}_n = A_n x_n + B_n u_n + E_n x_{n-1},
\end{equation}
where 
$x_n = [\tilde{h}_n ~ \tilde{v}_n ~ a_n]^\top$, 
$x_{n-1} = [\tilde{h}_{n-1} ~ \tilde{v}_{n-1} ~ a_{n-1}	]^\top$, and
\begin{align}
A_n = \begin{bmatrix}
	0 & -1 & 0 \\ 0 & 0 & 1 \\ 0 & 0 & -\frac{1}{\tau_n}
\end{bmatrix}, ~
B_n = \begin{bmatrix}
	0 \\ 0 \\ \frac{1}{\tau_n}
\end{bmatrix},~
E_n = \begin{bmatrix}
	0 & 1 & 0  \\ 0 & 0 & 0 \\ 0 & 0 & 0
\end{bmatrix}. \nonumber
\end{align}

The way that the HV $i$, $i \in [1,n-1]$, follows its preceding vehicle $i-1$ in the longitudinal direction is characterised by the model below (which is the standard OV model \cite{orosz2010traffic} augmented with actuator dynamics):
\begin{subequations} \label{HV dynamics}
	\begin{align}
		\dot{h}_i &= v_{i-1} - v_i, \\
		\dot{v}_i &= a_i, \\
		\dot{a}_i &= \frac{1}{\tau_i} \left[ \alpha_i (V(h_i) - v_i) + \beta_i (v_{i-1} - v_i) - a_i \right],
	\end{align}
\end{subequations}
where $h_i = p_{i-1} - p_i$. The scalars $\alpha_i$, $\beta_i$ and $\tau_i$ are the 
unknown headway gain, relative velocity gain and propulsion time constant, respectively. $V(h_i)$ is the desired velocity of HV $i$ which is represented as a piecewise function as follows: 
\begin{eqnarray} \label{range policy}
	V(h_i) = \left \{
	\begin{array}{cl}
		0, & h_i \leq h_\text{s} \\
		\frac{v_\text{max}}{2} \left[ 1 - \cos\left(\pi \frac{h_i - h_\text{s}}{h_\text{g} - h_\text{s}} \right) \right], & h_\text{s} < h_i < h_\text{g} \\
		v_\text{max}, & h_i \geq h_\text{g} 
	\end{array} \right.
\end{eqnarray}
with $h_\text{s}$ and $h_\text{g}$ being the smallest and largest spacing thresholds. 
This paper focuses on the case that a platoon is able to be established and thus $V(h_i)$ takes values under $h_\text{s} < h_i < h_\text{g}$. 

Suppose the reference velocity $v^* \leq v_\text{max}$ and AV 0 can track it accurately, then the HVs ultimately reach the set point $(h^*,v^*)$, where $h^*$ is computed from \eqref{range policy} and given as $h^* = \frac{h_g - h_s}{\pi} \cos^{-1}\left( 1 - \frac{2 v^*}{v_{\max}}\right) + h_s$. 
The platooning errors of HV $i$ are defined as $\tilde{h}_i = h_i - h^*$ and $\tilde{v}_i = v_i - v^*$, $i \in [1,n-1]$.
Linearising the HVs models around the set point $(h^*,v^*)$ gives the platooning error dynamics
\begin{equation} \label{HV error sys}
\dot{x}_i = A_i x_i + E_i x_{i-1},
\end{equation}
where
\begin{align}
&x_i = \begin{bmatrix}
	\tilde{h}_i \\ \tilde{v}_i \\ a_i
\end{bmatrix},~
A_i = \begin{bmatrix}
	0 & -1 & 0\\
	0 & 0 & 1 \\
	\frac{\bar{\alpha}_i}{\tau_i} & \frac{-\bar{\beta}_i}{\tau_i} & \frac{-1}{\tau_i}
\end{bmatrix},~
E_i = \begin{bmatrix}
	0 & 1 & 0 \\ 0& 0 & 0 \\ 0 & \frac{\bar{c}_i}{\tau_i} & 0
\end{bmatrix}, \nonumber\\
&\bar{\alpha}_i = \alpha_i \left.\frac{\partial V(h_i)}{\partial h_i} \right|_{h_i = h^*}, ~\bar{\beta}_i = \alpha_i + \beta_i, ~\bar{c}_i = \beta_i. \nonumber
\end{align}

Define $x = [x_1^\top, \cdots, x_n^\top]^\top$ as the overall platooning error vector, $u = u_n$ as the control input, and $w = [w_1, w_2]^\top$ as the disturbance, where $w_1 = v_0 - v^*$ represents the unknown deviation of the leader velocity $v_0$ from the reference velocity $v^*$ and $w_2 = \bar{c}_1 w_1/\tau_1$. 
The overall platooning error system is derived from \eqref{AV error sys} and \eqref{HV error sys} and given as
\begin{equation} \label{platoon error sys0}
	\dot{x} = A_c x + B_c u + D_c w
\end{equation}
with the system matrices
\begin{align}
	A_c &= \begin{bmatrix}
		A_1 &  &  &  \\
		E_2 & A_2 &  &  \\
		&  \ddots & \ddots & \\
		&         &   E_n  & A_n  
	\end{bmatrix}, 
	B_c = \begin{bmatrix}
		B_1 \\ B_2 \\ \vdots \\ B_n 
	\end{bmatrix}, 
	D_c = \begin{bmatrix}
		D_1 \\ D_2 \\ \vdots \\ D_n 
	\end{bmatrix}, \nonumber\\
	D_1 &= \begin{bmatrix} 1 & 0 \\ 0 & 0 \\ 0 & 1 \end{bmatrix}, B_i = \mathbf{0}, i \in [1,n-1], ~
	D_i = \mathbf{0}, i \in [2,n].
	\nonumber
\end{align} 

The forward Euler method is used to derive the following discrete-time system of \eqref{platoon error sys0} under a sampling time $t_s$: 
\begin{equation}\label{eq:sys for design}
x(k+1) = A x(k) + B u(k) + D w(k) 
\end{equation}
with the system matrices $A = I_{n_x} + t_s A_c$, $B = t_s B_c$ and $D = t_s D_c$. 
The dimensions of $x(k)$, $u(k)$, and $w(k)$ are denoted as $n_x = 3n$, $n_u = 1$, and $n_w = 2$, respectively. 

Due to the intrinsic uncertainty and randomness in human driving behaviours, it is difficult to known exactly the parameters $\alpha_i$ and $\beta_i$ in the OV model \eqref{HV error sys}. The propulsion time constants $\tau_i$, $i \in [1,n-1]$, are also unknown. Hence, the matrix $A$ in \eqref{eq:sys for design} is unknown and the existing model-based CACC strategies are not fit for the purpose. 

This paper will propose a data-driven dual-loop CACC strategy as outlined in Figure \ref{fig1} to obtain the controller $u(k)$ for AV $n$ based on \eqref{eq:sys for design}. The control objectives are as follows:
\begin{enumerate}
	\item[(i)] Stability of the mixed platoon, \textit{\textit{i.e.}}, steering the platooning error vector $x(k)$ to be zero. 
	\item[(ii)] Satisfaction of the following input and safety constraints: 
	\begin{equation} \label{control obj 2}
		u(k) \in \mathbb{U}, ~ C x(k) \in \mathbb{X},	
	\end{equation}
	where $C = [\mathbf{0}_{3 \times 3(n-1)}, ~ I_3]$. The constraint sets are defined as $\mathbb{U} = \{ u \in \mathbb{R} \mid |u| \leq u_\text{max} \}$ and 
	$\mathbb{X} = \{ x \in \mathbb{R}^{n_x} \mid |x| \leq x_\text{max} \}$ with the given bounds $x_\text{max} = [\tilde{h}_\text{max}~ \tilde{v}_\text{max}~ u_\text{max}]^\top$.  
\end{enumerate}
The inner loop constant-gain state feedback controller realises objective (i) and the outer loop MPC refines the inner loop controller to realise both objectives (i) and (ii).

The proposed design needs the following assumption:
\begin{assumption}\label{assume:disturbance}
	$|w| \leq \delta \times \mathbf{1}_{n_w}$ for a known scalar $\delta > 0$.	
\end{assumption} 

This assumption is reasonable because $w$ depends on $v_0$, $v^*$, $\beta_1$ and $\tau_1$ which are all 
physically bounded. 
Suppose that 
$|a_0| \leq u_{\max}$, 
$\underline{\beta}_1 \leq \beta_1 \leq \overline{\beta}_1$ and $\underline{\tau}_1 \leq \tau_1 \leq \overline{\tau}_1$, then one has
$$
|w_1| \leq u_{\max} t_s := \delta_1, ~
|w_2| = |\beta_1 w_1/\tau_1 | \leq \bar{\beta}_1 \delta_1/\underline{\tau}_1  := \delta_2. \nonumber
$$
The value of $\delta$ can be chosen as $\delta = \max\{ \delta_1, \delta_2\}$.

\section{Data-driven Inner Loop Control Design}\label{sec: inner loop}
This section develops a state feedback constant-gain controller based on the existing data-driven control theory \cite{DeTesi19,Van+20a,Van+20b,Persis+22}. Different from the literature, the obtained controller serves as the inner loop in the proposed dual-loop setting and lays the foundation for developing the data-driven 
outer loop MPC policy in next section to enforce the input/safety constraints. 

The inner loop control design begins with collecting the data of platooning errors.
At the start of forming the platoon (before the data-driven controller is designed and applied to AV $n$), AV $n$ uses the classic ACC controller \cite{Shladover+12} to maintain safety. A total number of $T$ data samples of platooning errors and control signals are collected and grouped as follows:
\begin{subequations}\label{eq:data seq1}	
\begin{align}
\setlength{\parindent}{0in}		
U_0 &= [u(0), u(1), \cdots, u(T-1)] \in \mathbb{R}^{n_u \times T}, \\
X_0 &= [x(0), x(1), \cdots, x(T-1)] \in \mathbb{R}^{n_x \times T}, \\
X_1 &= [x(1), x(2), \cdots, x(T)] \in \mathbb{R}^{n_x \times T}. 
\end{align}	
\end{subequations}
The data sequence $X_0$ is required to be of full row rank \cite{Persis+22}, which is necessary to ensure feasibility of the proposed data-driven inner loop controller design.

Let the sequence of unknown disturbance be
\begin{equation}\label{eq:data seq2}
W_0 = [w(0),w(1),\dots,w(T-1)] \in \mathbb{R}^{n_w \times T}.
\end{equation}

By using \eqref{eq:data seq1} and \eqref{eq:data seq2}, the platooning error system \eqref{eq:sys for design} can be replaced by its data-based alternative given in 
Lemma~\ref{lemma:data-based closed-loop sys}.
\begin{lemma}\label{lemma:data-based closed-loop sys}
If there are matrices $K \in \mathbb{R}^{n_u \times n_x}$ and $G \in \mathbb{R}^{T \times n_x}$ satisfying the equation
\begin{equation}\label{eq:lemma data closed 1}
\begin{bmatrix}
K \\ I_{n_x}
\end{bmatrix}	
=
\begin{bmatrix}
U_0 \\ X_0
\end{bmatrix}	
G,
\end{equation}
then applying the controller $u(k) = K x(k)$ to \eqref{eq:sys for design}  
results in the closed-loop platooning error system
\begin{equation}\label{eq:lemma data closed 2}
x(k+1) = \bar{A} x(k) + D w(k),
\end{equation}
where $\bar{A} = X_1 G - D W_0 G$.
\end{lemma}
\begin{proof}
Applying $u(k) = K x(k)$ to \eqref{eq:sys for design} and using 
\eqref{eq:lemma data closed 1} yields
\begin{align}\label{eq:lemma data closed pf1}
x(k+1) &= [B ~ A] \begin{bmatrix} K \\ I_{n_x} \end{bmatrix} x(k) + D w(k) \nonumber\\
&= (A X_0 + B U_0) G x(k) + D w(k).
\end{align}
Since $U_0$, $X_0$, $X_1$ and $D_0$ satisfy $x(s+1) = A x(s) + B u(s) + D w(s), ~s \in [0,T-1]$, the relation $X_1 = A X_0 + B U_0 + D W_0$ holds. Applying this to \eqref{eq:lemma data closed pf1} gives
\begin{equation}\label{eq:lemma data closed pf2}
x(k+1) = (X_1 - D W_0) G x(k) + D w(k).
\end{equation}
This immediately leads to \eqref{eq:lemma data closed 2} by defining $\bar{A} = X_1 G - D W_0 G$.
\end{proof}	
Since $X_0$ has full row rank $n_x$, there always exists a matrix $G$ such that $I_{n_x} = X_0 G$. Hence, the matrix $K = U_0 G$ always exists.
Based on Lemma~\ref{lemma:data-based closed-loop sys}, 
the proposed data-driven control is stated as below.
\begin{theorem}\label{thm:control design}
Under Assumption \ref{assume:disturbance},
the platooning error system \eqref{eq:sys for design} is asymptotically stable when $D w(k) = 0$ by implementing the controller $u(k) = K x(k)$ with
\begin{equation}\label{controller:thm}
	K = U_0 Y P^{-1},
\end{equation} 
if there exist matrices $P \in \mathbb{R}^{n_x \times n_x}$ and $Y \in \mathbb{R}^{T \times 
n_x}$ and a scalar $\gamma$ such that the following SDP problem is feasible:

\begin{subequations}\label{op:control design}	
\begin{align}
\label{const:control 1}
&\hspace{1.6cm} P \succ 0, ~ \gamma > 0, \\
\label{const:control 2}
&\hspace{1.6cm} X_0 Y = P, \\
\label{const:control 3}
&
\begin{bmatrix}
P - \gamma I_{n_x} & Y^\top X_1^\top & Y^\top & \mathbf{0}_{n_x \times n_w}\\
\star & P & \mathbf{0}_{n_w \times T} & D \Delta \\
\star & \star & \epsilon I_{T} & \mathbf{0}_{T \times n_w} \\
\star & \star & \star & \epsilon^{-1} I_{n_w}
\end{bmatrix} 
\succ 0,
\end{align}
\end{subequations}	 
where $\epsilon > 0$ is a user-specified scalar.
\end{theorem}
\begin{proof}
Suppose the SDP problem \eqref{op:control design} is feasible. Let $G = Y P^{-1}$, then the constraint \eqref{const:control 2} is equivalent to
\begin{equation}\label{eq:thm control pf1}
X_0 G = I_{n_x}.	
\end{equation}
Combining \eqref{controller:thm} with \eqref{eq:thm control pf1} gives \eqref{eq:lemma data closed 1}, which enables the use of Lemma \ref{lemma:data-based closed-loop sys} and leads to \eqref{eq:lemma data closed 2}. 

The next step is to prove that \eqref{const:control 3} ensures asymptotic stability of the system
\begin{equation}\label{eq:thm control pf2}
x(k+1) = (X_1 G - D W_0 G) x(k) 	
\end{equation}
where $W_0\in \mathcal{W} := \{ W \in \mathbb{R}^{n_w \times T} \mid |W(:,j)| \leq \delta \times \mathbf{1}_{n_w}, ~ j \in [1,T]\}$.

Consider the Lyapunov function $V(k) = x(k)^\top P^{-1} x(k)$. A sufficient condition ensuring asymptotic stability of \eqref{eq:thm control pf2} is 
$V(k+1) - V(k)  < 0$. 
Substituting \eqref{eq:thm control pf2} into this condition gives
\begin{align}\label{eq:thm control pf4}	
\setlength{\parindent}{0in}	
x(k)^\top \Pi x(k) < 0,
\end{align}
where $\Pi = (X_1 G - D W_0 G)^\top P^{-1} (X_1 G - D W_0 G) - P^{-1}$.

The condition \eqref{eq:thm control pf4} is equivalent to $\Pi \prec 0$. Multiplying both of its sides with $P$ and using $G P = Y$ leads to
\begin{equation}\label{eq:thm control pf5}	
Y^\top(X_1 - D W_0)^\top P^{-1} (X_1 - D W_0)Y - P \prec 0.
\end{equation}

To speed up the convergence to the origin, the condition \eqref{eq:thm control pf5} is replaced by a stronger one as follows:
\begin{equation}\label{eq:thm control pf6}	
	Y^\top(X_1 - D W_0)^\top P^{-1} (X_1 - D W_0)Y - P + \gamma I_{n_x} \prec 0,
\end{equation}
where $\gamma > 0$ is a decision variable.

Applying Schur complement to \eqref{eq:thm control pf6} and after some re-arrangement, it gives
\begin{equation}\label{eq:thm control pf13}
\begin{bmatrix}
P - \gamma I_{n_x} & Y^\top X_1^\top \\
\star & P
\end{bmatrix} 
-  M W_0^\top N - N^\top W_0 M^\top \succ 0,
\end{equation}	
where
$M^\top = \left[ Y, ~\mathbf{0}_{T \times n_x} \right]$ and
$N = \left[ \mathbf{0}_{n_w \times n_x}, ~D^\top\right].
$

It follows from \cite[Lemma 4]{Persis+22} that
\begin{equation}
M W_0^\top N + N^\top W_0 M \preceq \epsilon^{-1} M M^\top + \epsilon N^\top \Delta \Delta^\top N,	
\end{equation}
for any given scalar $\epsilon > 0$ and $\Delta = \delta \sqrt{T} I_{n_w}$. Hence, a sufficient condition for \eqref{eq:thm control pf13} is given as
\begin{equation}\label{eq:thm control pf14}
\begin{bmatrix}
	P - \gamma I_{n_x} & Y^\top X_1^\top \\
	\star & P
\end{bmatrix}  - \epsilon^{-1} M M^\top - \epsilon N^\top \Delta \Delta^\top N \succ 0.	
\end{equation}
Apply Schur complement to \eqref{eq:thm control pf14} gives \eqref{const:control 3}.
In summary, \eqref{const:control 3} is sufficient for ensuring \eqref{eq:thm control pf5} and thus asymptotic stability of the system \eqref{eq:thm control pf2}.
\end{proof}	

A necessary condition for ensuring feasibility of the SDP problem \eqref{op:control design} is that $X_0$ has full row rank \cite{Persis+22}, because it ensures 
\eqref{const:control 2} admit a solution. 
The controller determined from Theorem \ref{thm:control design} only ensures asymptotic stability of the platooning error system when $D w(k) = 0$. To ensure platoon stability when $D w(k) \neq 0$ and satisfaction of the control input and safety constraints specified in \eqref{control obj 2}, the controller applied to the AV $n$ is refined as 
\begin{equation}\label{eq:refined controller}
	u(k) = K x(k) + \hat{u}(k),
\end{equation}
where $\hat{u}(k)$ is a the outer loop MPC controller designed in Section \ref{sec: outer loop}. 

\section{Outer Loop MPC Design}\label{sec: outer loop}
The platooning error system model under the refined controller \eqref{eq:refined controller}
is derived as
\begin{subequations}\label{observer_eq:data closed 1}
\begin{align}
x(k+1) &= \bar{A} x(k) + B \hat{u}(k) + D d(k),	\\
y(k) &= x(k),
\end{align}	
\end{subequations}
where $\bar{A} = X_1 G_1$ and $d(k) = - W_0 G x(k) + w(k)$ represents 
the lumped effect of $w(k)$ and the unknown coupling of $x(k)$ and $W_0$. $y(k)$ is the output measurement.

At the steady state, AV 0 tracks $v^*$ accurately (\textit{i.e.}, $v_0 = v^*$) and $x(k) = \mathbf{0}$, thus $w = \mathbf{0}$ and the lumped disturbance $d(k)$ vanishes. This inspires us to design the outer loop MPC following the concept of offset free MPC \cite{Maeder+09} which consists in using an internal model to predict future system behaviour. 

\subsection{Data-driven internal model}
This paper builds an internal model as the data-based nominal system $x(k+1) = \bar{A} x(k) + B \hat{u}(k)$ augmented with a disturbance model that captures the mismatch between the nominal and real platooning error system~\eqref{observer_eq:data closed 1}. 
Since the lumped disturbance $d(k)$ depends on $x(k)$ and $w(k)$ and subsequently the vehicle position, velocity and acceleration, a second-order linear disturbance model is used to derive the following internal model:
\begin{align}\label{observer_eq:aug model}
	\underbrace{\begin{bmatrix}
		x(k+1) \\ \omega_1(k+1) \\ \omega_2(k+1)
\end{bmatrix}}_{\xi(k+1)} &= 
\underbrace{\begin{bmatrix}
		\bar{A} & B_d & \mathbf{0} \\
		\mathbf{0} & I_{n_w} & t_s I_{n_w} \\
		\mathbf{0} & \mathbf{0} & I_{n_w}
\end{bmatrix}}_{A_\xi} \underbrace{\begin{bmatrix}
		x(k) \\ \omega_1(k) \\ \omega_2(k)
\end{bmatrix}}_{\xi(k)}  + 
\underbrace{\begin{bmatrix}
		B \\ \mathbf{0} \\ \mathbf{0}
\end{bmatrix}}_{B_\xi} \hat{u}(k), \nonumber\\
y(k) &= \underbrace{\begin{bmatrix}
		I_{n_x} ~ C_d ~ \mathbf{0}
\end{bmatrix}}_{C_\xi} \xi(k),
\end{align} 
where $w_1(k) = w(k)$ and $w_2(k) \in \mathbb{R}^{n_w}$. The constant matrices $B_d , C_d \in \mathbb{R}^{n_x \times n_w}$ are chosen such that 
\begin{equation}\label{rank constraint}
\setlength\arraycolsep{1pt}	
\mathrm{rank}	
\begin{bmatrix}
\lambda I_{n_x + 2 n_w} - A_\xi \\ 
C_\xi
\end{bmatrix} = n_x + 2 n_w,
\end{equation}
for all $\lambda$ with $\mathrm{Re}(\lambda) \geq 0$.

At each step $k$, the MPC will use \eqref{observer_eq:aug model} to predict the future platooning error dynamics, which requires the initial internal model state value $\xi(k)$. Since $\omega_1(k)$ and $\omega_2(k)$ are unknown, $\xi(k)$ is unknown. To address this, an observer can be used to estimate its value. 
The fulfilment of \eqref{rank constraint} ensures that the augmented system~\eqref{observer_eq:aug model} is detectable \cite{LanPatton20} and there always exists a stable observer to estimate $\xi(k)$ accurately.

The proposed observer takes the form of
\begin{subequations}\label{observer_eq:data observer2}
	\begin{align}
		z(k+1) &= N_\xi z(k) + G_\xi \hat{u}(k) + L y(k), \\
		\hat{\xi}(k) &= z(k) + H_\xi y(k),
	\end{align}		
\end{subequations}
where $z(k) \in \mathbb{R}^{n_x+2n_w}$ is the observer state and 
$\hat{\xi}(k)$ is the estimate of $\xi(k)$. The matrices $N_\xi \in 
\mathbb{R}^{(n_x+2n_w) \times (n_x+2n_w)}$, $G_\xi \in \mathbb{R}^{(n_x+2n_w) 
	\times n_u}$, $L \in \mathbb{R}^{(n_x+2n_w) \times n_x}$ 
and $H_\xi \in \mathbb{R}^{(n_x+2n_w) \times n_x}$ are to be designed.

Define the estimation error as $e(k) = \xi(k) - \hat{\xi}(k)$. It can be 
derived that $e(k) = \Phi \xi(k) - z(k)$, where $\Phi = I - H_\xi 
C_\xi$. The estimation error dynamics are then derived as
\begin{align}\label{observer_eq:error2 sys}
	e(k+1) ={}&  N_\xi e(k) + (\Phi A_\xi - N_\xi - L_1 C_\xi) \xi \nonumber\\
	&+ (\Phi B_\xi - G_\xi) \hat{u} + (N_\xi H_\xi - L_2) y(k),
\end{align}
where $L = L_1 + L_2$.

Define the following equations:
\begin{align}\label{observer_eq:matrix conds}
	& \Phi A_\xi - N_\xi - L_1 C_\xi = \mathbf{0},~ 
	\Phi B_\xi - G_\xi = \mathbf{0}, \nonumber\\
	 & N_\xi H_\xi - L_2 = \mathbf{0}. 
\end{align}	
Applying \eqref{observer_eq:matrix conds} to the estimation error dynamics \eqref{observer_eq:error2 sys} gives
\begin{align}\label{observer_eq:error2 sys2}
	e(k+1) = (\Phi A_\xi - L_1 C_\xi) e(k). 
\end{align}

The way of designing the gains $H_\xi$ and $L_1$ to stabilise the estimation error system \eqref{observer_eq:error2 sys2} is presented in Lemma
\ref{lemma:observer}.
\begin{lemma}\label{lemma:observer}
The estimation error system~\eqref{observer_eq:error2 sys2} is asymptotically stable if there 
exists a matrix $P_o \in \mathbb{R}^{(n_x+2n_w) \times (n_x+2n_w)}$ and a scalar $\epsilon_o$ such that the following SDP problem is feasible:
\begin{subequations}\label{lemma_observer:sdp}
\begin{align}	
	\label{lemma_const: 1}
	& \hspace{2cm} P_o \succ 0,~ \epsilon_o > 0, \\	
	\label{lemma_const: 2}
	& \begin{bmatrix}
		P_o - \epsilon_o I & (P_o A_\xi - \bar{H} C_\xi A_\xi - \bar{L}_1 C_\xi)^\top \\
		\star & P_o 
	\end{bmatrix} \succ 0.
\end{align}	
\end{subequations}
The gains are 
obtained as $H_\xi = P_o^{-1} \bar{H}$ and $L_1 = P_o^{-1} \bar{L}_1$.	
\end{lemma}
\begin{proof}
Consider a Lyapunov function $V_o(k) = e(k)^\top P_o^{-1} e(k)$ with $P_o \succ 0$.
The matrix $(\Phi A_\xi - L_1 C_\xi)$ is Schur stable if 
\begin{equation}
V_o(k+1) - V_o(k) < 0,	
\end{equation}
and equivalently, 
\begin{equation}\label{lemma_pf:observer eq1}
	(\Phi A_\xi - L_1 C_\xi)^\top P_o^{-1} (\Phi A_\xi - L_1 C_\xi) - P_o^{-1} \prec 0. 
\end{equation}	
Multiplying both sides of \eqref{lemma_pf:observer eq1} with $P_o$ and defining $\bar{H} = H_\xi P_o$ and $\bar{L}_1 = L_1 P_o$, it then gives
\begin{equation}\label{lemma_pf:observer eq2}
\Omega^\top P_o^{-1} \Omega - P_o  \prec 0,
\end{equation}
where $\Omega = P_o A_\xi - \bar{H} C_\xi A_\xi - \bar{L}_1 C_\xi$.

To improve convergence of the observer, we consider the following condition that is stronger than \eqref{lemma_pf:observer eq2}:
\begin{equation}\label{lemma_pf:observer eq3}
	\Omega^\top P_o^{-1} \Omega - P_o + \epsilon_o I_{n_x + 2 n_w} \prec 0,
\end{equation}
where $\epsilon_o > 0$ is a decision variable.
Applying Schur complement to \eqref{lemma_pf:observer eq3} gives \eqref{lemma_const: 2}. 
\end{proof}	

Solving \eqref{lemma_observer:sdp} gives $H_\xi$ and $L_1$, which are submitted into \eqref{observer_eq:matrix conds} to get $N_\xi$, $G_\xi$, $L_2$, and thus $L = L_1 + L_2$. Convergence of the observer can be further improved by placing the poles of $(\Phi A_\xi - L_1 C_\xi)$ to a desirable region \cite{KrokavecFilasova13} through adding an extra linear constraint into \eqref{lemma_observer:sdp}.

\subsection{MPC design} \label{prob original}
The proposed MPC design is formulated as the optimisation problem:
\begin{subequations}\label{op:MPC}	
\begin{align}
& \min\limits_{c(0),\cdots,c(N-1)} J_N \nonumber\\
\label{mpc st 1}	
\text{s.t.} ~ &x(j+1) = \bar{A} x(j) + B c(j) + B_d w_1(j),  j \in [0,N],\\
\label{mpc st 2}
 & w_1(j+1) = w_1(j) + t_s w_2(j), ~j \in [0,N], \\
 & w_2(j+1) = w_2(j), ~j \in [0,N],\\
\label{mpc st 3}	
 & C x(j) \in \mathbb{X}, ~ j \in [0,N], \\
 &c(j) + K x(j) \in \mathbb{U}, ~ j \in [0,N-1], \\
 &x_0 = \hat{x}(k), ~ w_1(0) = \hat{w}_1(k), ~ w_2(0) = \hat{w}_2(k),
\end{align}
\end{subequations}
where the cost function is defined as $J_N = \|x(N) - \bar{x}(k) \|^2_P 
+ \sum_{j=0}^{N-1} (\|x(j) - \bar{x}(j) \|^2_Q + \|c(j) - \bar{u}(k) \|^2_R)$.
The weights $Q$ and $R$ satisfy $Q \succeq 0$ and $R \succ 0$, and $P$ is solved from the Riccati equation
$
P = \bar{A}^\top P \bar{A} - (\bar{A}^\top P B) (B^\top P B + R)^{-1} (B^\top P \bar{A}) + Q.
$
The targeted steady-state state $\bar{x}(k)$ and input $\bar{u}(k)$ are solved from the following optimisation problem:
\begin{subequations}\label{op:steady-state}
\begin{align}
& \min\limits_{\bar{x}(k), \bar{u}(k)} \| \bar{x}(k)\|^2_{\bar{Q}} + \| \bar{u}(k) \|^2_{\bar{R}} \nonumber\\	
\text{s.t.} ~& 
\label{op:steady-state cond1}
	\begin{bmatrix}
	I_{n_x} - \bar{A} & -B	\\C & \mathbf{0} 
	\end{bmatrix}
	\begin{bmatrix}\bar{x}(k) \\ \bar{u}(k) \end{bmatrix}
	= 
	\begin{bmatrix}
	B_d \hat{d}(k) \\ -C_d \hat{d}(k)
	\end{bmatrix}, \\
\label{op:steady-state cond2}
	& C \bar{x}(k) \in \mathbb{X}, ~ \bar{u}(k) \in \mathbb{U},
\end{align}
\end{subequations}
where $\bar{Q} \succ 0$ and $\bar{R} \succeq 0$ are the given weights.

Solving the MPC problem~\eqref{op:MPC} gives the optimal control sequence $\{c^*(j)\}_{j=0}^{N-1}$ and the obtained MPC policy is set as $\hat{u}(k) = c^*(0)$.
Property of the proposed overall control law is discussed in Theorem \ref{thorem:mpc}. 
\begin{theorem}\label{thorem:mpc}
Suppose the matrices $B_d$ and $C_d$ are chosen to satisfy \eqref{rank constraint} and the optimisation problems \eqref{op:control design}, \eqref{lemma_observer:sdp}, \eqref{op:MPC} and \eqref{op:steady-state} are feasible. Applying the overall control law $u(k) = K x(k) + \hat{u}(k)$ to AV $n$ ensures stability of the mixed platoon under the input and safety constraints in \eqref{control obj 2}.	
\end{theorem}
\begin{proof}
If \eqref{lemma_observer:sdp} is feasible, then the estimation error system \eqref{observer_eq:error2 sys2} is asymptotically stable, \textit{i.e.}, $\lim_{k \rightarrow \infty} e(k) = \mathbf{0}$ and $\lim_{k \rightarrow \infty} e_y(k) = \lim_{k \rightarrow \infty} (y(k) - C_\xi \hat{\xi}(k)) = \mathbf{0}$. It is shown in \cite{LanPatton20} that the observer \eqref{observer_eq:data observer2} can be reformulated as the following Proportional-Derivative Luenberger observer:
\begin{equation}\label{thorem:mpc pf1}
	\hat{\xi}(k+1) = \bar{A} \hat{\xi}(k) + B \hat{u}(k) + \hat{L} e_y(k)  + \hat{H} e_y(k+1),
\end{equation} 
where $\hat{L} = \Phi^{-1} L_1$, $\hat{H} = \Phi^{-1} H$ and $e_y(k) = y(k) - C_\xi \hat{\xi}(k)$.

At steady state (\textit{i.e.}, $k \rightarrow \infty$), $e_y(k) = e_y(k+1) = \mathbf{0}$ and $y_\infty - C_\xi \hat{\xi}_\infty = \mathbf{0}$.  
It then follows from \eqref{thorem:mpc pf1} that
\begin{equation}\label{thorem:mpc pf2}
	\hat{\xi}_\infty = \bar{A} \hat{\xi}_\infty + B \hat{u}_\infty.
\end{equation} 
Hence, the steady state of the observer \eqref{observer_eq:data observer2} satisfy
\begin{equation}\label{thorem:mpc pf3}
\begin{bmatrix}
	I_{n_x} - \bar{A} & -B	\\ C & \mathbf{0} 
\end{bmatrix}
\begin{bmatrix}\bar{x}_\infty \\ \bar{u}_\infty \end{bmatrix}
= 
\begin{bmatrix}
	B_d \hat{d}_\infty \\ y_\infty -C_d \hat{d}_\infty
\end{bmatrix}.
\end{equation}

To ensure platoon stability, the controller $u(k) = K x(k) + \hat{u}(k)$ needs to steer the platooning error $x(k)$, \textit{i.e.} $y(k)$, to zero. 
This means that at steady state $y_\infty = \mathbf{0}$.
Hence, the equation \eqref{thorem:mpc pf3} becomes
\begin{equation}\label{thorem:mpc pf4}
	\begin{bmatrix}
		I_{n_x} - \bar{A} & -B	\\ C & \mathbf{0} 
	\end{bmatrix}
	\begin{bmatrix}\bar{x}_\infty \\ \bar{u}_\infty \end{bmatrix}
	= 
	\begin{bmatrix}
		B_d \hat{d}_\infty \\ -C_d \hat{d}_\infty
	\end{bmatrix}.	
\end{equation}
In view of this, solving the optimisation problem \eqref{op:steady-state} gives the target steady state $\bar{x}$ and input $\bar{u}$ that can ensure platoon stability (\textit{i.e.}, \eqref{op:steady-state cond1}) and satisfaction of the safety/input constraints (\textit{i.e.}, \eqref{op:steady-state cond2}). Furthermore, 
feasibility of the MPC problem \eqref{op:MPC} ensures the platooning error state $x(k)$ and input $u(k)$ track their targeted values $\bar{x}$ and $\bar{u}$, guaranteeing platoon stability and satisfaction of the safety/input constraints. 
\end{proof}	

The string stability, in the sense of head-to-tail string stability \cite{Lan+21d}, is a direct consequence of Theorem \ref{thorem:mpc}. This is because the proposed dual-loop control for the AV $n$ at the rear ensures stability of the platooning error dynamics and satisfaction of the constraints in \eqref{control obj 2}. These constraints include the input constraint $|u| \leq u_{\max}$ and the safety constraints on the relative distance $|\tilde{h}_n| < \tilde{h}_{\max}$ and the relative velocity $|\tilde{v}_n| \leq \tilde{v}_{\max}$ between AV $n$ and its preceding vehicle HV $n-1$. The fulfilment of these safety constraints means that the rear AV can safely follow the preceding HVs whenever the lead vehicle velocity changes propagate downstream the entire platoon. The head-to-tail string stability of the established mixed platoon will also be demonstrated in the simulation results.

\section{Simulation Results} \label{sec:simulation}
Simulation validation of the proposed design is performed on a mixed platoon with six vehicles, including two AVs at the front and rear, and four HVs in between. 
The simulation use the following vehicle parameters: $\tau_1 = 0.13$, $\alpha_1 = 0.2$, $\beta_1 = 0.4$,
$\tau_2 = 0.12$, $\alpha_2 = 0.2$, $\beta_2 = 0.45$,
$\tau_3 = 0.16$, $\alpha_3 = 0.3$, $\beta_3 = 0.4$, 
$\tau_4 = 0.15$, $\alpha_4 = 0.2$, $\beta_4 = 0.45$, 
$\tau_5 = 0.12$. The other parameters are: $h_s = 5~\mathrm{m}$, $h_g = 50~\mathrm{m}$, $v_{\max} = 40~\mathrm{m/s}$, $u_{\max} = 4~\mathrm{m/s^2}$, $t_s = 0.05~\mathrm{s}$, and $d_\mathrm{safe} = 20~\mathrm{m}$. The initial vehicle state $(p_i,v_i,a_i)$, $i \in [0,5]$, are randomly set as: 
$(120,20,0)$, $(100,20,0)$, $(80,15,0)$, $(65,20,0)$, $(40,20,0)$ and $(15,15,0)$, respectively. 

Simulations are conducted in MATLAB with the SDP problems being solved via YALMIP \cite{lofberg2004yalmip} and MOSEK \cite{mosek2010mosek}. Comparisons are made for the proposed 
data-driven dual-loop control method, the classic ACC \cite{Shladover+12} and the existing single-loop data-driven MPC \cite{Lan+21d}. 
For the proposed method, the number of data samples is set as $T = 500$, the disturbance model uses the matrices $B_d = \mathbf{1}_{15 \times 2}$ and 
$C_d = [0, 1; \mathbf{1}_{14 \times 2}]$, and the parameters used for solving the optimisation problems \eqref{op:MPC} and \eqref{op:steady-state} are $N = 2$, $Q = 10\times I_{15}$, $R = 1$, $\bar{Q} = 10^3 \times I_{15}$, $\bar{R} = 0$. The gains of the classic ACC controller are borrowed from the MATLAB example ``Adaptive Cruise Control with Sensor Fusion''. The data-driven MPC follows the same settings as the reference \cite{Lan+21d}.

\begin{figure}[t]
	\centering
	\includegraphics[width=\columnwidth]{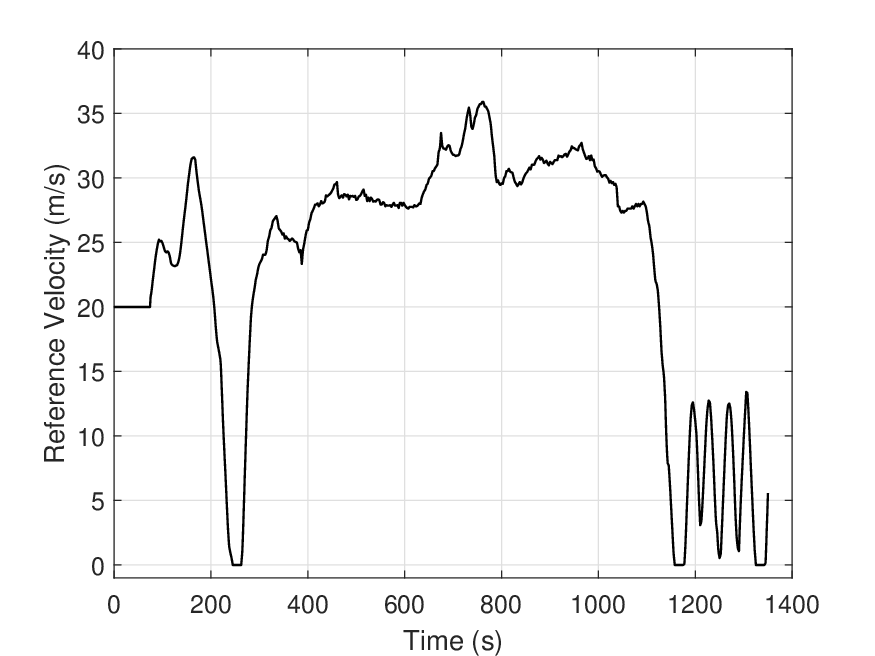} 
	\vspace{-7mm} 
	\caption{Reference velocity.}
	\label{fig2}
\end{figure}

\begin{figure}[t]
	\centering
	\includegraphics[width=\columnwidth]{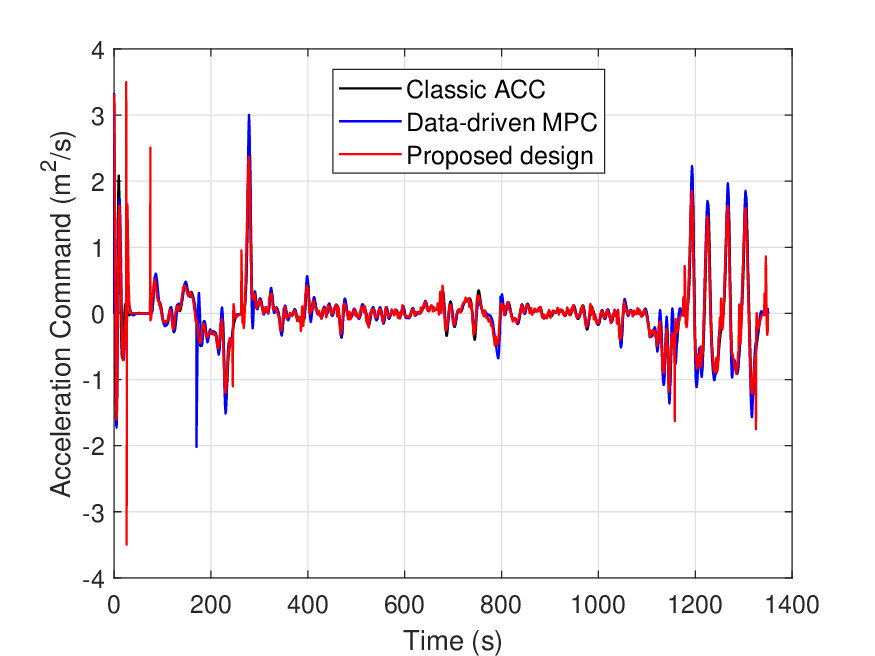}  
	\vspace{-7mm}
	\caption{Acceleration commands of AV 5.}
	\label{fig3}
\end{figure}

The leader AV 0 uses the MPC controller in the work \cite{lan2020min} to track the velocity profile in Figure \ref{fig2}, which combines 
a 75~s constant speed driving at $20~\mathrm{m/s}$ and the SFTP-US06 Drive 
Cycle. This velocity profile can well validate the practical effectiveness of the proposed design, because it consists of high speed driving, high acceleration and rapid speed fluctuations.
The data-driven MPC \cite{Lan+21d} needs $170~\mathrm{s}$ to collect enough data, while the proposed design needs $25~\mathrm{s}$.
It takes an average runtime of $0.004~\mathrm{s}$ to solve the proposed MPC, which is much less than the $0.03~\mathrm{s}$ needed to solve the data-driven MPC \cite{Lan+21d}.

The control commands of AV 5 reported in Figure \ref{fig3} show that the input limit $[-4, 4]~\mathrm{m}/\mathrm{s}^2$ is satisfied by all the methods. Compared to the other two methods, the proposed method enables AV 5 to have smaller velocity deviations, as shown in Figure \ref{fig4}. The gap between HV 4 and AV 5 under the three designs are reported in Figure \ref{fig5}. The classic ACC has large changes in the inter-vehicular distance due to the use of a time-varying desired spacing policy. The proposed method and data-driven MPC aim to maintain a constant desired spacing of $20~\mathrm{m}$ between HV 4 and AV 5. Compared to data-driven MPC, the proposed method is better in maintaining the targeted inter-vehicular gap, with quick response and smaller overshoots.

\begin{figure}[t]
	\centering
	\includegraphics[width=\columnwidth]{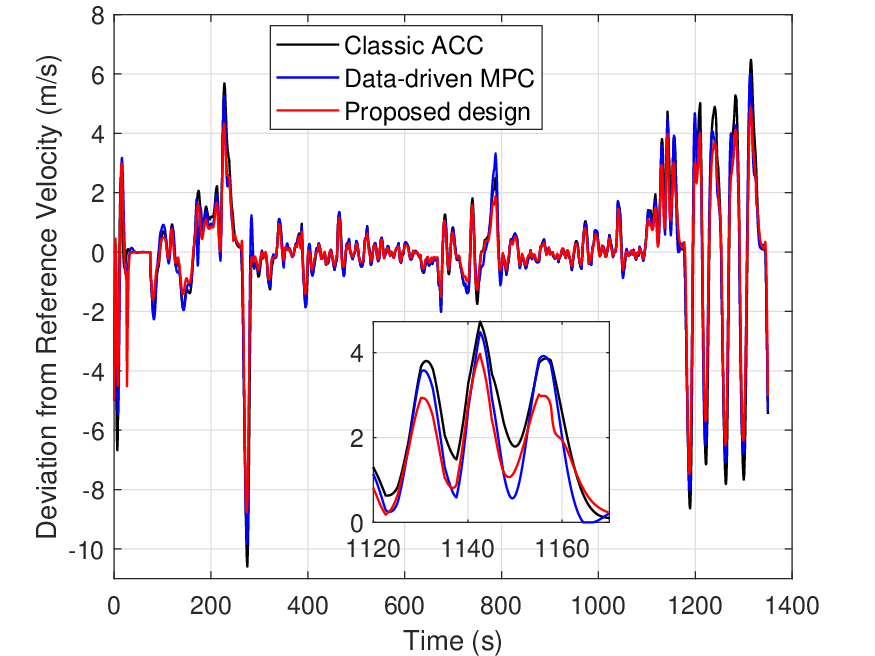}  
	\vspace{-7mm}
	\caption{Deviations of AV 5's velocity from the reference.}
	\label{fig4}
\end{figure}

\begin{figure}[t]
	\centering
	\includegraphics[width=\columnwidth]{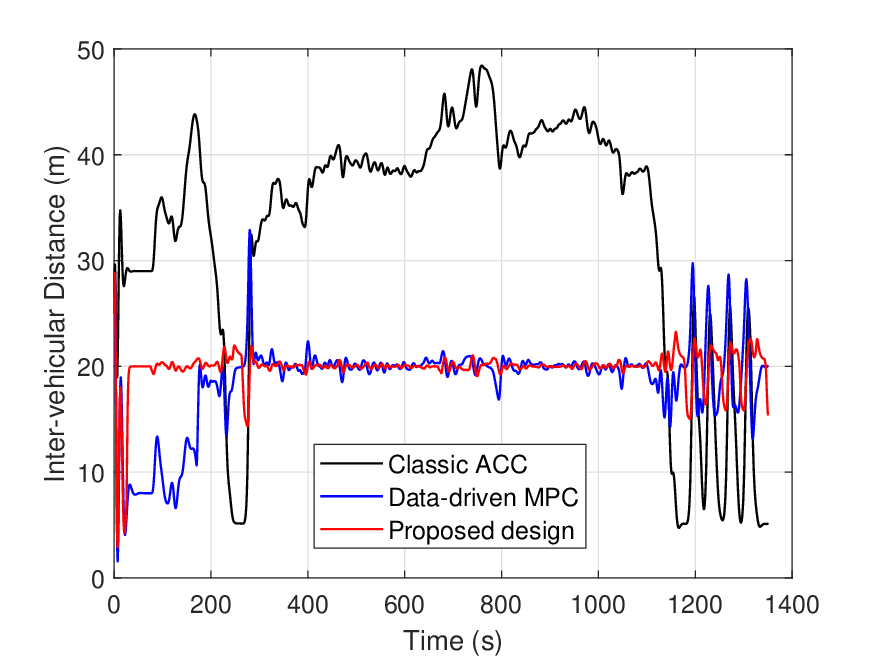} 
	\vspace{-7mm} 
	\caption{Spacing between HV 4 and AV 5.}
	\label{fig5}
\end{figure}

\begin{figure}[t]
	\centering
	\includegraphics[width=\columnwidth]{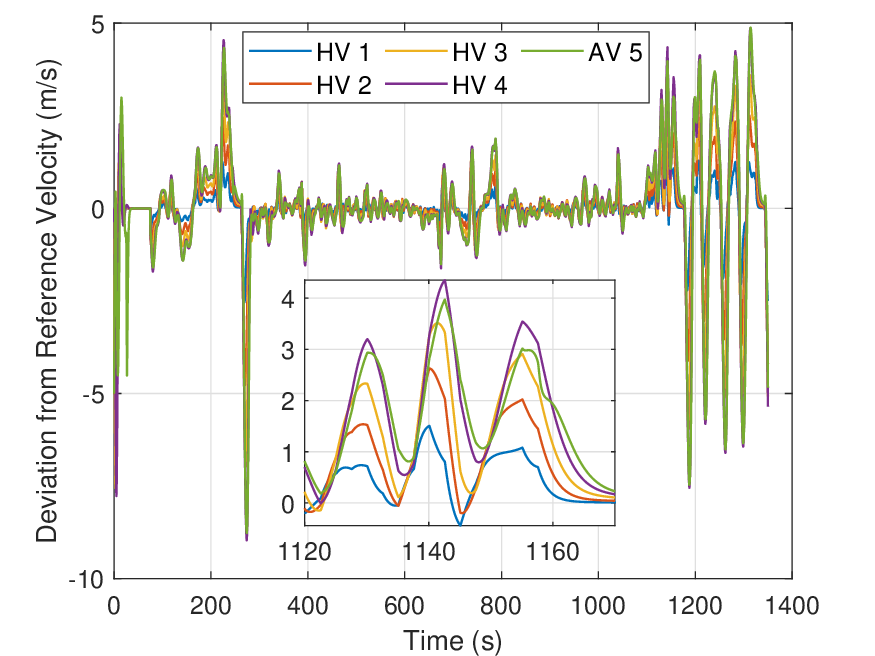} 
	\vspace{-7mm} 
	\caption{Velocity deviations under the proposed design.}
	\label{fig6}
\end{figure}

The deviations of all vehicles' velocities from the reference velocity are used to  check whether the proposed CACC can make the mixed vehicle platoon string stable. As shown in Figure \ref{fig6}, the velocity deviation is amplified when propagating from HV 1 all the way down to HV4, but attenuated by AV 5. This confirms that the proposed design can guarantee head-to-tail string stability \cite{Lan+21d} of the mixed vehicle platoon.

\section{Conclusion} \label{sec:conclusion}
This paper proposes a data-driven dual-loop CACC for mixed vehicle platoons subject to unknown HV model parameters and propulsion time constants. 
The inner loop constant-gain state feedback controller ensures platoon stability, and the outer loop MPC refines the inner loop controller to ensure both platoon stability and satisfaction of safety/input constraints. 
The simulation 
results demonstrate that the proposed CACC is more effective than the classic ACC and existing data-driven MPC methods in ensuring mixed platoon stability under a representative aggressive velocity reference. The proposed method is also  computationally cheaper than the existing data-driven MPC. 
Future research will be extending the proposed CACC for platoons with more realistic nonlinear vehicle dynamics and noisy sensor measurements.





\bibliographystyle{WileyNJD-AMA}
\bibliography{References}

\end{document}